\documentclass[conference, A4paper, 10pt]{IEEEtran}

\ifCLASSINFOpdf
\else
   \usepackage[dvips]{graphicx}
\fi
\usepackage{url}
\usepackage{graphicx}
\usepackage{balance}
\usepackage{stfloats}
\usepackage{fancyhdr}
\usepackage{multirow, cite}
\usepackage{xcolor,colortbl, soul}
\usepackage{tabularx,booktabs}
\usepackage[utf8]{inputenc} 
\usepackage[T1]{fontenc}
\usepackage{url}
\usepackage{ifthen}
\usepackage{cite}
\usepackage[cmex10]{amsmath} % Use the [cmex10] option to ensure 
\usepackage{multicol,lipsum}
\usepackage{amsmath,amssymb,amsfonts,amsthm,bm}
\usepackage{subfigure}
\usepackage{cite}
\usepackage{amsmath,amssymb,graphicx,epsfig,fancyhdr,amsthm,tabulary,epstopdf}
\usepackage{algorithm}%algorithmic   algorithm
\usepackage{algorithmic}
\usepackage{pdfsync}
\usepackage{bm}
\usepackage{caption}
\newtheorem{thm}{Theorem}
\newtheorem{cor}{Corollary}

\newtheorem{pro}{Proposition}

\theoremstyle{remark}
\newtheorem{rem}{Remark}

\theoremstyle{definition}

\newcommand{\CASE}[1]{\STATE \textbf{case} #1\textbf{:} \begin{ALC@g}}
	\newcommand{\ENDCASE}{\end{ALC@g}}

\newcommand{\DEFAULT}{\STATE \textbf{default:} \begin{ALC@g}}
	\newcommand{\ENDDEFAULT}{\end{ALC@g}}
\newcommand{\DEFAULTLINE}[1]{\STATE \textbf{default:} }

\newcounter{MYtempeqncnt}

\begin{document}

\title{
Optimal Order of Encoding for Gaussian MIMO Multi-Receiver Wiretap Channel
 %Secure Weighted Sum Rates of MIMO Multi-Receiver Wiretap Channels
}

%%% Several authors with up to three affiliations:
\author{%
	\IEEEauthorblockN{Yue Qi and Mojtaba Vaezi}
	\IEEEauthorblockA{
		Department of Electrical and Computer Engineering, 		Villanova University, Villanova, PA 19085, USA\\
		Email:  \{yqi, mvaezi\}@villanova.edu}
	
}

\maketitle

\begin{abstract}
The Gaussian multiple-input multiple-output (MIMO) multi-receiver wiretap channel is studied in this paper. The base station broadcasts confidential messages to $K$  intended users while keeping the messages secret from an eavesdropper. The capacity of this channel has already been characterized by applying dirty-paper coding and stochastic encoding. However, $K$ factorial  encoding orders may need to be enumerated for that, which makes the
problem intractable. We prove that there exists one optimal encoding order and reduced the  $K$ factorial times to a one-time encoding. The optimal encoding order is proved by forming a secrecy weighted sum rate (WSR) maximization problem.  The optimal order is the same as that for the MIMO broadcast channel  without
secrecy constraint, that is, the weight of users' rate in the WSR
maximization problem determines the optimal encoding order.
Numerical results verify the optimal encoding order. 
\end{abstract}

\IEEEpeerreviewmaketitle

\section{Introduction}

The \textit{multi-receiver wiretap channel}, depicted in Fig.~\ref{fig:system}, is a channel model in which a transmitter wants to transmit messages of $K$ legitimate users while keeping them confidential from an external eavesdropper.  This model, which is an extension of the well-known wiretap channel \cite{wyner1975wire}, is also known as secure broadcasting \cite{ekrem2012degraded} and  wiretap broadcast channel  \cite{benammar2015secrecy}.   
 In the multiple-input multiple-output (MIMO) multi-receiver wiretap channel each node can have   an arbitrary number of antennas  \cite{ekrem2011secrecy}. 
  The secrecy capacity region of the two-receiver channel was characterized in \cite{liu2010vector, bagherikaram2013secrecy}     %aligned MIMO-BC wiretap channels, 
 %(\textit{degraded} eavesdropper compared with a \textit{more-noisy} one in \cite{ekrem2011secrecy}) 
 % over additive white Gaussian noise (AWGN)
 in which secret dirty-paper coding (S-DPC)  is proved to be optimal \cite{bagherikaram2013secrecy}. 
 %An outer bound  of the capacity region for a general eavesdropper is shown in \cite{benammar2015secrecy}. 

Ekrem and Ulukus established the secrecy capacity region of the Gaussian MIMO $K$-receiver wiretap channel in \cite{ekrem2011secrecy}. Interestingly, a different order of encoding for users can result in a different achievable region. % for this channel.  
The capacity-achieving  rate region is then characterized by the convex closure of the union of achievable regions obtained by all  one-to-one permutations of encoding order.  Therefore,    to determine
the entire capacity region, all possible one-to-one permutations of users, i.e., $K!$  possible encoding orders, may need to be enumerated which makes the
problem intractable especially when the number of users becomes large.

In this paper, we prove that such  enumerations can be avoided and we find the optimal encoding order. To this end, we form a weighted sum-rate (WSR) maximization problem for the $K$-receiver wiretap channel,  convert the objective function to another equivalent function, and prove the optimal order for the new optimization problem. 
% Important applications of WSR formulation arise from achieving the capacity of MIMO-BC \cite{liu2008maximum}, stabilizing the transmission
% buffers to guarantee fairness for downlinks \cite{kobayashi2006iterative}. 
Our work shows that the optimal encoding order for  secure broadcasting is the same
 as that for  insecure broadcasting, i.e., the MIMO-BC without 
 secrecy. More specifically, the descending weight ordering in the WSR
 maximization problem determines the optimal encoding order \cite{liu2008maximum}. The proof is non-trivial because, unlike the MIMO-BC, the  associated problem in the multi-receiver MIMO wiretap channel is non-convex and BC-MAC duality \cite{vishwanath2003duality} cannot be directly applied. 
 %Our proof shows that the problem can be transformed to its MAC wiretap channel \cite{bagherikaram2010secure, xie2013secure}. 
 
 This finding reduces the complexity of evaluating the WSR problem by $K$ factorial. 
 When all weights are different, there is only one optimal order of encoding. However,  
 when some weights are equal, different encoding orders can give different rate tuples (corner points of the capacity region) while all of them result in the same WSR. 
 Particularly, encoding order is not important when sum-capacity is the concern, i.e., all weights are the same. Nonetheless, the order will determine which corner point of the capacity region  will be achieved.   
 
\begin{figure}[tb]
	\centering
	\includegraphics[width=0.43\textwidth]{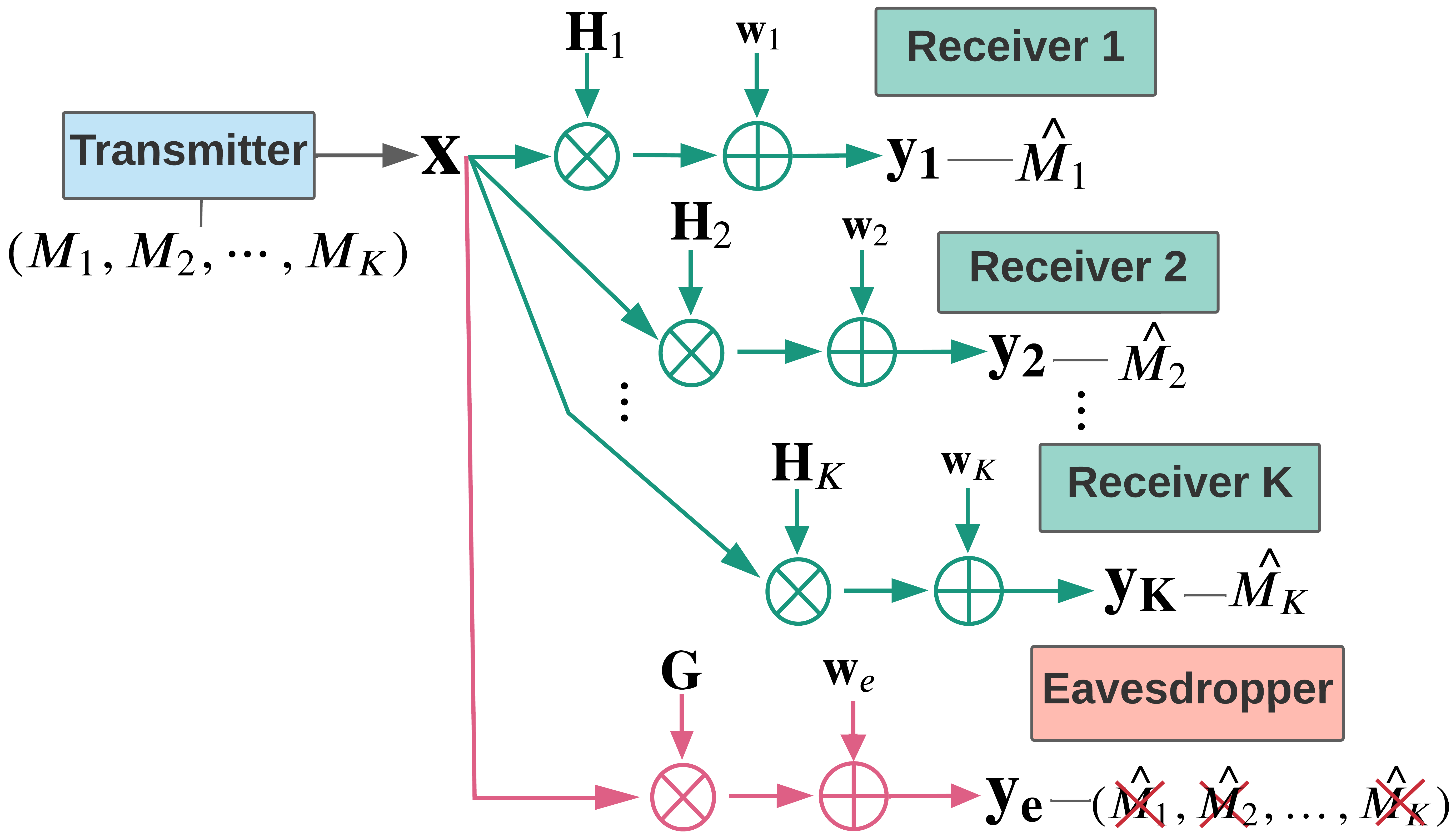}
	% where an .eps filename suffix will be assumed under latex,
	% and a .pdf suffix will be assumed for pdflatex
	\caption{System model for the $K$-receiver wiretap channel.}
	\label{fig:system}
\end{figure}

 In addition to the order of encoding, we provide a numerical solution for finding the covariance matrices that achieve the capacity region. Previously,  the solution existed only for the secrecy sum-capacity of this channel \cite{park2016secrecy}, but not the whole capacity region. 
 Simulation results confirm the optimality of the proposed encoding order.
 
% by applying multiple access channel (MAC) to broadcast channel (BC) duality \cite{vishwanath2003duality}. 

 The remainder of this paper is organized as follows. In Section~\ref{sec:II}, we describe the system model.
 We introduce the WSR maximization problem and give the proof for the optimal weight ordering in Section~\ref{sec:III}. We provide a numerical solution to efficiently find the covariance matrices in Section~\ref{sec:IV}. Then we provide examples to verify the results and conclude the paper in Section~\ref{sec:V} and Section~\ref{sec:VI}.
 
\textit{Notations:} $\rm{tr}(\cdot)$ and $(\cdot)^{\dagger}$ denote trace and Hermitian of matrices. $\mathbb{E}\{\cdot\}$ denotes expectation.  $\mathbf{Q} \succcurlyeq \mathbf{0} $ represents that $\mathbf{Q}$ is a positive semidefinite matrix. $\mathbf{I}_n$ is the identity matrix with size $n \times n$.  $\log$ represents  the natural logarithm and $|\mathbf{Q}|$ represents the determinant of $\mathbf{Q}$. %$Re\{\}$

%It is worth noting that we classify the {Scenarios} from how the information theoretic results have been established, and} 

\section{System Model} \label{sec:II}
Consider a $K$-receiver wiretap channel, as shown in 
Fig.~\ref{fig:system}. The BS transmits independent  confidential messages to $K$ legitimate receivers in the presence of one eavesdropper. Each confidential message $M_k$ for user $k$, $k= 1, 2, \dots, K$, should be kept secret  from the eavesdropper. There is no cooperation  among legitimate users.  The BS, user $k$, and  the eavesdropper are equipped with $n_t$, $n_k$, and $n_e$ antennas, respectively.   
The received signals at the legitimate users and at the eavesdropper are given by \vspace{-3pt}
\begin{subequations}\label{eq:signal model}
	\begin{align} 
		\mathbf{y}_k  &= \mathbf{H}_k  \mathbf{x} +\mathbf{w}_k, \; k= 1, 2, \dots, K, \\
		\mathbf{y}_e &= \mathbf{G}\mathbf{x} + \mathbf{w}_e,
	\end{align}
\end{subequations}
\noindent where the channels of $k$th legitimate user $\mathbf{H}_k$  and  eavesdropper $\mathbf{G}$ are ${n_k \times n_t}$ and ${n_e \times n_t}$ complex matrices  in which the elements of the channels  are drawn from  independent and identically distributed (i.i.d.) complex Gaussian distributions  and  $\mathbf{w}_k$, and  $\mathbf{w}_e$ are complex random vectors whose elements are i.i.d., zero-mean, unit-variance Gaussian. 
%following  $\mathcal{CN}(0, \mathbf{I})$, and 
The $K$ independent signals are encoded by Gaussian codebooks and superimposed, and Gaussian signaling is optimal \cite{ekrem2011secrecy}. Thus, the transmit signal is $\mathbf{x}=\sum_{k=1}^{K}\mathbf{x}_k$ where $\mathbf{x}_k \sim \mathcal{CN}(0, \mathbf{Q}_k)$.  $\mathbf{Q}_k\succcurlyeq  \mathbf{0}$ represents the  covariance matrix of $k$th user.

The channel input is subject to a power constraint. 
While it is more common to define capacity regions
under a total power constraint, i.e., ${\rm tr}(\mathbb{E}\{\mathbf{x}\mathbf{x}^{\dagger}\}) \leq P$, the capacity region of this channel is first proved based on a covariance matrix constraint, i.e., $\mathbb{E}\{\mathbf{x}^{\dagger}\mathbf{x}\} \preccurlyeq \mathbf{S}$  \cite{ekrem2011secrecy, ekrem2012degraded}.
Once the capacity region is obtained under a covariance constraint $\mathcal{C}(\{\mathbf{H}_k\}_{k=1}^{K}, \mathbf{G}, \mathbf{S})$,  the capacity region
under more relaxed constraints on the channel inputs  $\mathcal{C}(\{\mathbf{H}_k\}_{k=1}^{K}, \mathbf{G}, P)$ can be obtained by union over compact sets of input covariance matrices as $\mathcal{C}(\{\mathbf{H}_k\}_{k=1}^{K}, \mathbf{G}, P) = \bigcup \limits_{\mathbf{S} \succcurlyeq \mathbf{0}, {\rm tr}(\mathbf{S}) \leq P} \mathcal{C}(\{\mathbf{H}_k\}_{k=1}^{K}, \mathbf{G}, \mathbf{S}).$ 
%\begin{align} \label{constraintP}
%\mathcal{C}(\{\mathbf{H}_k\}_{k=1}^{K}, \mathbf{G}, P) = \bigcup \limits_{\mathbf{S} \succcurlyeq \mathbf{0}, {\rm tr}(\mathbf{S}) \leq P} \mathcal{C}(\{\mathbf{H}_k\}_{k=1}^{K}, \mathbf{G}, \mathbf{S}).
%\end{align}
 In short, the the channel input should satisfy  \vspace{-3pt}
\begin{align} \label{eq:powercons} 
%\frac{1}{n}\sum_{t=1}^{n}||\mathbf{x}||^2  \leq P
{\rm tr}(\mathbb{E}\{\mathbf{x}\mathbf{x}^{\dagger}\}) = \sum_{k=1}^{K}{\rm tr}(\mathbf{Q}_k) \leq {\rm tr}(\mathbf{S}) \leq P.
\end{align}
%Matrix  covariance constraint is also considered in \cite{ekrem2011secrecy, ekrem2012degraded} $\mathbb{E}\{\mathbf{x}^{\dagger}\mathbf{x}\} \preccurlyeq \mathbf{S}$
%%\begin{align} \label{eq:powercov cons}
%%\frac{1}{n}\sum_{t=1}^{n}||\mathbf{x}||^2  \leq P
%%\mathbb{E}\{\mathbf{x}^{\dagger}\mathbf{x}\} \preccurlyeq \mathbf{S},
%%\end{align}
%in which $\mathbf{S}$ is a strictly positive definite matrix. 

%Thus, the capacity regions over total power constraint  can be described by compact sets of input covariance matrices
	
Let $\pi = [\pi_1, \pi_2, \dots, \pi_K]$ be a permutation function on the set $\{1, \dots, K\}$ and $\pi_k=j$ represent the $k$th element of that arrangement is $j$. By applying the DPC and stochastic
encoding, the achievable non-negative secrecy rate at user $\pi_k$ under the encoding order $\pi$ is given by \cite{ekrem2011secrecy, park2016secrecy}
\begin{align}\label{eq_KUrate}
	R_{\pi_k}=&\log 
	\frac{\left|
		\mathbf{I}_{n_k}+\mathbf{H}_{\pi_k}
		\left(\sum_{j=k}^{K}\mathbf{Q}_{\pi_j}\right)\mathbf{H}_{\pi_k}^\dagger
		\right|}{\left|
		\mathbf{I}_{n_k}+\mathbf{H}_{\pi_k}
		\left(\sum_{j=k+1}^{K}\mathbf{Q}_{\pi_j}\right)\mathbf{H}_{\pi_k}^\dagger
		\right|} \notag \\
	& \quad \quad - \log 
	\frac{\left|
		\mathbf{I}_{n_e}+\mathbf{G}
		\left(\sum_{j=k}^{K}\mathbf{Q}_{\pi_j}\right)\mathbf{G}^\dagger
		\right|}{\left|
		\mathbf{I}_{n_e}+\mathbf{G}
		\left(\sum_{j=k+1}^{K}\mathbf{Q}_{\pi_j}\right)\mathbf{G}^\dagger
		\right|},
\end{align}
in which $k=1, 2,  \dots, K$. Since $\log$ represents the natural logarithm, the unite of rates is {${\rm nats/sec/Hz}$}, in this paper. The secrecy capacity region under the total power constraint \eqref{eq:powercons} is characterized by the convex closure of the union of the DPC rate region over all possible one-to-one permutations $\pi$ \cite[Theorem 4]{ekrem2011secrecy}, where $\pi_k = j$ represents the $k$th position in permutation $\pi$ is user $j$. In general, $K!$ possible  encoding orders may need  enumerating to determine
the capacity region.   %which is given by

It is worth mentioning that when there is only one legitimate receiver, the problem reduces to the MIMO wiretap
	channel and \eqref{eq_KUrate} becomes the capacity expression
	of the MIMO Gaussian wiretap channel,  established in  \cite{khisti2010secure, oggier2011secrecy, liu2009note}.
Also, various linear precoding schemes are designed for this case ($K=1$), including generalized singular value decomposition \cite{fakoorian2012optimal}, alternating optimization and water filling \cite{li2013transmit}, and rotation modeling \cite{vaezi2017journal, zhang2020rotation}. In this paper, we are interested in the $K$-receiver wiretap channel for $K \geq 2$, but our solutions apply to the case with  $K =1$. 
\section{Optimal Encoding Order}\label{sec:III}
%we assume the encoding order is an ascending order, that is, the message of user 1 is first encoded, and then the message of user 2, and so on. 
The WSR maximization for the $K$-receiver wiretap channel under a total power constraint $P$ is formulated as 
\begin{align} \label{eq_WSR}
	\varphi(P)=&\max \limits_{\mathbf{Q}_{\pi_k} \succcurlyeq\mathbf{0}}
	\sum_{k=1}^K w_{\pi_k} R_{{\pi_k}}, \quad k=1,2, \dots, K \notag \\
	&{\rm s.t.} \quad \sum_{k=1}^K{\rm tr}(\mathbf{Q}_{\pi_k}) \leq P,
\end{align}
in which $R_{k}$ is given in \eqref{eq_KUrate}. The weights $w_{\pi_1}, w_{\pi_2}, \dots, w_{\pi_K}$ are non-negative values adding to one. Solving  problem  \eqref{eq_WSR} for {different permutations $\pi$ results in}  different inner bounds (achievable regions)  of the secrecy capacity region.
The whole capacity region is then determined by enumerating $K!$ possible encoding orders and finding their union, which makes the
problem intractable if the number of users is large.\footnote{Unlike this, in the MIMO-BC without security, there exists an optimal encoding order $\pi$ to maximize the WSR problem \cite{tse1998multiaccess, liu2008maximum}. The proof is determined by the well-known duality between BC and MAC \cite{vishwanath2003duality}.  
}  	

%
%The Lagrangian of the problem  \eqref{eq_WSR} is
%\begin{align} \label{lagr}
%L(\mathbf{Q}_1, \mathbf{Q}_2,\dots, \mathbf{Q}_K, \lambda) =
%\sum_{k=1}^K w_k R_{k}-\lambda(\sum_{k=1}^K{\rm tr}(\mathbf{Q}_k)-P),
%\end{align}
%where $\lambda$ is the Lagrange multiplier related to the total power constraint. The dual function is a maximization of the Lagrangian 
%\begin{align} \label{dual}
%g(\lambda)  = \max \limits_{\mathbf{Q}_k \succcurlyeq\mathbf{0}} L(\mathbf{Q}_1, \mathbf{Q}_2,\dots, \mathbf{Q}_K, \lambda),
%\end{align}
%and the dual problem is given by
%\begin{align} \label{dualp}
%\min \limits_{\lambda \geq 0}g(\lambda){ .}
%\end{align}
%Usually, if the objective functions and the constraints are convex, standard
%convex optimization results guarantee that the primal problem \eqref{lagr} and the dual problem \eqref{dualp} have the same solution. For example, the MIMO-BC (without eavesdropping) \cite{weingartens2006capacity} satisfies Slater’s condition and KKT conditions
%are sufficient for optimality \cite{liu2008maximum, qi2021signaling}. When
%convexity does not hold, there may exist a duality gap between the primal and dual problems. However, duality gap can be zero and the KKT conditions {will be} necessary even when the optimization problem is not convex \cite{yu2006dual, park2015weighted}. We have the following lemma. 

%
%\begin{lem} \label{lemma1}
%	The problem in \eqref{eq_WSR} has zero duality gap and the KKT conditions are necessary  for the optimal solution.  
%\end{lem}
%
%\textit{proof:} See the proof in Appendix~\ref{proofLemma}.

In this section, 
%we prove that in the $K$-user wiretap channel,  only one encoding order is enough to achieve every point on the boundary  of the capacity region. In other words, 
we prove that once the weights are fixed, the encoder order in \eqref{eq_WSR} is determined. 
Since the proof uses the BC-MAC duality \cite{vishwanath2003duality}, before providing the proof, we first briefly introduced this result  in the following. 

	\begin{pro} (BC-MAC Duality \cite{vishwanath2003duality})
	Suppose that users $1,  \dots, K$ are encoded sequentially. The  achievable rate of {user~$k$ for the MIMO-BC can be computed by}
		%, the capacity region of a MIMO-BC is given in \cite{vishwanath2003duality}.
		\begin{align}
			R^{\rm BC}_{k} = \log 
		\frac{\left|
			\mathbf{I}_{n_k}+\mathbf{H}_{k}
			\left(\sum_{j=k}^{K}\mathbf{Q}_{j}\right)\mathbf{H}_{k}^\dagger
			\right|}{\left|
			\mathbf{I}_{n_k}+\mathbf{H}_{k}
			\left(\sum_{j=k+1}^{K}\mathbf{Q}_{j}\right)\mathbf{H}_{k}^\dagger
			\right|}, \label{eq:BCsum}
		\end{align} 
	and, in {its} dual MIMO-MAC, the achievable  rate is given by
		\begin{align}\label{mac}
		R^{\rm MAC}_{k} = \log 
		\frac{\left|
			\mathbf{I}_{n_t}+	\sum_{j=1}^{k}\mathbf{H}^\dagger_{j}
		\mathbf{\Sigma}_{j}\mathbf{H}_{j}
			\right|}{\left|
			\mathbf{I}_{n_t}+\sum_{j=1}^{k-1}\mathbf{H}^\dagger_{j}
		\mathbf{\Sigma}_{j}\mathbf{H}_{j}
			\right|}.
		\end{align}
\noindent Given a set of BC covariance matrices $\mathbf{Q}_{k}$, $k=1, \dots, K$, there are  MAC covariance matrices	$\mathbf{\Sigma}_{k}$, $k=1, \dots, K$, such that 	%the BC achievable rates are identical to the MAC 	achievable rates, 
	$R^{\rm MAC}_{k}=R^{\rm BC}_{k}$, $\forall k$, and $\sum_{k=1}^{K}{\rm tr}(\mathbf{Q}_{k})=\sum_{k=1}^{K}{\rm tr}(\mathbf{\Sigma}_{k})$ via the BC-MAC transformation, and vice versa. Let us define 
	\begin{align}
\mathbf{C}_{k}& \triangleq \mathbf{I}_{n_k}+\mathbf{H}_{k}
	(\sum_{j=k+1}^{K}\mathbf{Q}_{j})\mathbf{H}_{k}^\dagger, \notag \\
\mathbf{D}_{k}&  \triangleq	\mathbf{I}_{n_t}+\sum_{j=1}^{k-1}\mathbf{H}^\dagger_{j}
	\mathbf{\Sigma}_{j}\mathbf{H}_{j}.
	\end{align} 
	 Then, the BC and MAC covariance matrices can be expressed as \cite{vishwanath2003duality}
	\begin{align}
	\mathbf{Q}_{k} &= \mathbf{D}^{-1/2}_{k}\mathbf{E}_{k}\mathbf{F}^{\dagger}_{k}\mathbf{C}^{1/2}_{k}	\mathbf{\Sigma}_{k} \mathbf{C}^{1/2}_{k}\mathbf{F}_{k}\mathbf{E}^{\dagger}_{k}\mathbf{D}^{-1/2}_{k},	 \notag \\
	\mathbf{\Sigma}_{k} & =	\mathbf{C}^{-1/2}_{k}\mathbf{E}_{k}\mathbf{F}^{\dagger}_{k}\mathbf{D}^{1/2}_{k}	\mathbf{Q}_{k} \mathbf{D}^{1/2}_{k}\mathbf{F}_{k}\mathbf{E}^{\dagger}_{k}\mathbf{C}^{-1/2}_{k},
\end{align}	
in which  $\mathbf{E}_{k}$ and $\mathbf{F}_{k}$ are obtained by singular value decomposition of  $\mathbf{D}_{k}^{-1/2}\mathbf{H}^{\dagger}_ {k}\mathbf{C}_{k}^{-1/2}=\mathbf{E}_{k}\mathbf{\Lambda}_{k}\mathbf{F}_{k}^{\dagger}$, where $\mathbf{\Lambda}_{k}$ is a
	square and diagonal matrix.	
	\end{pro}

	\begin{thm} \label{lemma2}
		The WSR problem in \eqref{eq_WSR} can be solved by the following equivalent optimization problem
		\begin{subequations} \label{theorem}
		\begin{align} \label{eq_WSR1}
		&\max \limits_{\mathbf{\Sigma}_{\pi_k} \succcurlyeq\mathbf{0}}
		\sum_{k=1}^K (w_{\pi_k} - w_{\pi_{k-1}}) \times \notag \\
		& \bigg(\log|\mathbf{I}_{n_t} + \sum_{j=k}^{K}\mathbf{H}^{\dagger}_ {\pi_j}\mathbf{\Sigma}_ {\pi_j}\mathbf{H}_ {\pi_j}| - \log|\mathbf{I}_{n_e} + \sum_{j=k}^{K}\mathbf{G}_{\pi_j}^{\dagger}\mathbf{\Sigma}_ {\pi_j}\mathbf{G}_{\pi_j}\bigg), 	 \\ \label{eq:equavalent}
		&{\rm s.t.} \quad \sum_{k=1}^K{\rm tr}(\mathbf{\Sigma}_{\pi_k}) \leq P, \quad k=1,2, \dots, K 
		\end{align}
		\end{subequations}%\overset{\Delta}
		 {where $w_{\pi_0}\triangleq0$ and $\mathbf{G}_{\pi_j}\triangleq \mathbf{C}_{\pi_j}^{1/2}\mathbf{F}_{\pi_j}\mathbf{E}^{\dagger}_{\pi_j}\mathbf{D}_{\pi_j}^{-1/2}\mathbf{G}^{\dagger}$. Further, the optimal} decoding order $\pi$ is a permutation of the set $\{1,  \dots, K\}$ such that the weights satisfy $w_{\pi_1} \leq w_{\pi_2} \leq  \dots \leq w_{\pi_K}$. The optimal encoding order is the reverse.
		 %, i.e., $w_{\pi_1} \geq w_{\pi_2} \geq  \dots \geq w_{\pi_K}$.
	\end{thm}
\begin{proof}
		First,  by expanding and rewriting the WSR maximization problem \eqref{eq_WSR}, we have
		\begin{subequations}
		\begin{align}
		&\quad \max \limits_{\mathbf{Q}_{\pi_k} \succcurlyeq\mathbf{0}}	\sum_{k=1}^K w_{\pi_k} R_{\pi_k}  \notag \\
		&=	\max \limits_{\mathbf{Q}_{\pi_k} \succcurlyeq\mathbf{0}}\bigg( \sum_{k=1}^K  w_{\pi_k} \log 
		 \frac{\left|
			\mathbf{I}_{n_k}+\mathbf{H}_{\pi_k}
			\left(\sum_{j=k}^{K}\mathbf{Q}_{\pi_j}\right)\mathbf{H}_{\pi_k}^\dagger
			\right|}{\left|
			\mathbf{I}_{n_k}+\mathbf{H}_{\pi_k}
			\left(\sum_{j=k+1}^{K}\mathbf{Q}_{\pi_j}\right)\mathbf{H}_{\pi_k}^\dagger
			\right|}  \notag \\
		&   - \sum_{k=1}^K w_{\pi_k} \log 
		\frac{\left|
			\mathbf{I}_{n_e}+\mathbf{G}
			\left(\sum_{j=k}^{K}\mathbf{Q}_{\pi_j}\right)\mathbf{G}^\dagger
			\right|}{\left|
			\mathbf{I}_{n_e}+\mathbf{G}
			\left(\sum_{j=k+1}^{K}\mathbf{Q}_{\pi_j}\right)\mathbf{G}^\dagger
			\right|}\bigg)  \\
		& \stackrel{(a)}{=} \max \limits_{\mathbf{\Sigma}_{\pi_k} \succcurlyeq\mathbf{0}} \sum_{k=1}^K (w_{\pi_k} - w_{\pi_{k-1}})\log|\mathbf{I}_{n_t} + \sum_{j=k}^{K}\mathbf{H}^{\dagger}_ {\pi_j}\mathbf{\Sigma}_{\pi_j}\mathbf{H}_ {\pi_j}| \notag \\ 
		&   + \max \limits_{\mathbf{Q}_{\pi_k} \succcurlyeq\mathbf{0}} \bigg( - \sum_{k=1}^K w_{\pi_k} \log 
		\frac{\left|
			\mathbf{I}_{n_e}+\mathbf{G}
			\left(\sum_{j=k}^{K}\mathbf{Q}_{\pi_j}\right)\mathbf{G}^\dagger
			\right|}{\left|
			\mathbf{I}_{n_e}+\mathbf{G}
			\left(\sum_{j=k+1}^{K}\mathbf{Q}_{\pi_j}\right)\mathbf{G}^\dagger
			\right|}\bigg)  \label{bc_wsr}\\
			&\stackrel{(b)}{=} \max \limits_{\mathbf{\Sigma}_{\pi_k} \succcurlyeq\mathbf{0}} \sum_{k=1}^K (w_{\pi_k} - w_{\pi_{k-1}}) \log|\mathbf{I}_{n_t} + \sum_{j=k}^{K}\mathbf{H}^{\dagger}_ {\pi_j}\mathbf{\Sigma}_ {\pi_j}\mathbf{H}_ {\pi_j}| \notag \\
			& + \max \limits_{\mathbf{Q}_{\pi_k} \succcurlyeq\mathbf{0}}\bigg( -\sum_{k=1}^K(w_{\pi_k} - w_{\pi_{k-1}}) \log|\mathbf{I}_{n_e} + \mathbf{G}\sum_{j=k}^{K}\mathbf{Q}_{\pi_j}\mathbf{G}^{\dagger}|\bigg)  \label{mac_wsr} \\
			&\stackrel{(c)}{=} 	\max \limits_{\mathbf{\Sigma}_{\pi_k} \succcurlyeq\mathbf{0}} \sum_{k=1}^K (w_{\pi_k} - w_{\pi_{k-1}}) \times  \notag  \\ 
			& \hspace{-2.5mm} \bigg(\log|\mathbf{I}_{n_t} + \sum_{j=k}^{K}\mathbf{H}^{\dagger}_ {\pi_j}\mathbf{\Sigma}_ {\pi_j}\mathbf{H}_ {\pi_j}| - \log|\mathbf{I}_{n_e} + \sum_{j=k}^{K}\mathbf{G}_{\pi_j}^{\dagger}\mathbf{\Sigma}_ {\pi_j}\mathbf{G}_{\pi_j}|\bigg). \label{lasteq}
		\end{align}
		\end{subequations}
	Here, $(a)$ holds due to \cite[Theorem 1]{liu2008maximum}; $(b)$ holds by expanding the second term in \eqref{bc_wsr} and 
	rearranging the terms to obtain the second term in \eqref{mac_wsr}; and, $(c)$ is obtained by applying the BC-MAC duality and $\mathbf{G}_{\pi_j}\triangleq \mathbf{C}_{\pi_j}^{1/2}\mathbf{F}_{\pi_j}\mathbf{E}^{\dagger}_{\pi_j}\mathbf{D}_{\pi_j}^{-1/2}\mathbf{G}^{\dagger}$. 
	%We note that, 	the term in the second line of \eqref{lasteq} can be interpreted as   $K$-user MIMO MAC channel with an external eavesdropper \cite{bagherikaram2010secure, xie2013secure}.

 So far, we have proved that \eqref{theorem} is equivalent to \eqref{eq_WSR}. We next prove the optimal order of encoding.  If we prove that
 the difference of the two logs in \eqref{lasteq} is non-negative for any $\Sigma_{\pi_k}$,   the optimal order will be determined \cite{liu2008maximum}.
 \footnote{We assume that not all the weights are identical, otherwise, the WSR maximization reduces to a scaled sum-rate problem  whose optimal solution is  obtained	at any of the $K!$ corner points of the capacity region \cite{park2016secrecy}.} 
	
To this end, knowing  that $R_{\pi_k}\ge 0$ for every $\pi_k$, from \eqref{eq_KUrate} and  \eqref{eq:BCsum} we can check that 
\begin{subequations}
	\begin{align}
	&	R^{\rm BC}_{\pi_K} \geq
	\log{|
		\mathbf{I}_{n_e}+\mathbf{G}
		\mathbf{Q}_{\pi_K}\mathbf{G}^\dagger
		|},  \\
	%&	R^{\rm BC}_{K-1} + R^{\rm BC}_{K}  \geq
	%\log{|\mathbf{I}+\mathbf{G}
	%	(\sum_{j={K-1}}^{K}\mathbf{Q}_{\pi_j})\mathbf{G}^\dagger|}, \notag \\
		& \quad \quad  \quad \quad \quad \quad \vdots \notag \\
		& R^{\rm BC}_{\pi_k} +  \dots + R^{\rm BC}_{\pi_K}  \geq
		\log{\bigg|\mathbf{I}_{n_e}+\mathbf{G}
			(\sum_{j=k}^{K}\mathbf{Q}_{\pi_j})\mathbf{G}^\dagger \bigg|}. \label{bcinequal} 
%		\\
%		& \quad \quad  \quad \quad \quad \quad \vdots \notag \\
%		& R^{\rm BC}_{\pi_1} +  \dots + R^{\rm BC}_{\pi_K}  \geq
%	\log{|	\mathbf{I}_{n_e}+\mathbf{G}
%		(\sum_{j=1}^{K}\mathbf{Q}_{\pi_j})\mathbf{G}^\dagger|},  
	\end{align}
\end{subequations} 
Next, we have \begin{align}
\log|\mathbf{I} + \sum_{j=k}^{K}\mathbf{H}^{\dagger}_{\pi_j}\mathbf{\Sigma}_ {\pi_j}\mathbf{H}_{\pi_j}| & \stackrel{(a)}{\ge}   \log 
\frac{\left|
	\mathbf{I}+	\sum_{j=k}^{K}\mathbf{H}^\dagger_{\pi_j}
	\mathbf{\Sigma}_{\pi_j}\mathbf{H}_ {\pi_j}
	\right|}{\left|
	\mathbf{I}+\sum_{j=1}^{k-1}\mathbf{H}^\dagger_{\pi_j}
	\mathbf{\Sigma}_{\pi_j}\mathbf{H}_{\pi_j}
	\right|} \notag \\
&  \stackrel{(b)}{=} \sum_{j=k}^{K} R^{\rm MAC}_{\pi_j} 
%\notag
%\\
%&  
\stackrel{(c)}{=} \sum_{j=k}^{K} R^{\rm BC}_{\pi_j}  \notag \\
& \stackrel{(d)}{\ge}
\log{|	\mathbf{I}+\mathbf{G}
	(\sum_{j=k}^{K}\mathbf{Q}_{\pi_j})\mathbf{G}^\dagger|}  \notag \\
&   \stackrel{(e)}{=}
\log|\mathbf{I} + \sum_{j=k}^{K}\mathbf{G}_{\pi_j}^{\dagger}  \mathbf{\Sigma}_{\pi_j}\mathbf{G}_{\pi_j}|
\notag
\end{align} 
where the $(a)$ follows because  $\sum_{j=1}^{k-1}\mathbf{H}^{\dagger}_{\pi_j}\mathbf{\Sigma}_{\pi_j} \mathbf{H}_{\pi_j} \succcurlyeq \mathbf{0}$ which results in $\left|
\mathbf{I}+\sum_{j=1}^{k-1}\mathbf{H}^\dagger_{\pi_j}
\mathbf{\Sigma}_{\pi_j}\mathbf{H}_{\pi_j}
\right| \ge 1$, $(b)$ and $(c)$ follow from Proposition~1, $(d)$ holds because of \eqref{bcinequal}, and $(e)$ follows with the same argument we had moving from  \eqref{mac_wsr} to  \eqref{lasteq}.

This proves that in \eqref{lasteq}, the term inside the parentheses is non-negative. Then, to assure a non-negative WSR, we must have $w_{\pi_k} - w_{\pi_{k-1}} \geq 0$. Hence, $w_{\pi_1} \leq  w_{\pi_{2}} \leq  \dots \leq w_{\pi_{K}}$, and the decoding order is a permutation $[\pi_{1}, \pi_{2}, \dots, \pi_{K}]$ satisfying $w_{\pi_1} \leq  w_{\pi_{2}} \leq  \dots \leq w_{\pi_{K}}$. 
This completes the proof.
\end{proof}

\begin{rem}
The secrecy capacity region of the general Gaussian MIMO multi-receiver wiretap channel is characterized by the convex closure of the union of the DPC rate region over all possible one-to-one permutations $\pi$ \cite[Theorem 4]{ekrem2011secrecy}. Based on Theorem~\ref{lemma2}, only one encoding order is enough to achieve {each point on the boundary of} the secrecy capacity region, and $K!$  encoding orders is reduced to one order.  The optimal order is the same as that for the MIMO-BC, that is, a user	{with a higher} weight in the WSR should be encoded earlier.
\end{rem}

\begin{cor}
	When {all} weights are equal, the secrecy sum-rate is obtained. Different encoding orders may result in different rate tuples $(R_{\pi_1}, R_{\pi_2}, \dots, R_{\pi_K})$, but all encoding orders yield the same secrecy sum-rate. 
\end{cor}

\section{Solving the WSR Problem} \label{sec:IV}

\begin{figure*}[!b]
	% ensure that we have normalsize text
	\normalsize
	% Store the current equation number.
	\setcounter{MYtempeqncnt}{\value{equation}}
	% Set the equation number to one less than the one
	% desired for the first equation here.
	% The value here will have to changed if equations
	% are added or removed prior to the place these
	% equations are referenced in the main text.
	\setcounter{equation}{15}
	\begin{subequations}
		\begin{align}	
		\hline
		&f_k^{ccv}(\mathbf{Q}_k,\mathbf{Q}_{\bar{k}})=w_k 
		\log{|
			\mathbf{I}+({
				\mathbf{I}+\mathbf{H}_k
				\sum_{j=k+1}^{K}\mathbf{Q}_j\mathbf{H}_k^\dagger
			})^{-1}\mathbf{H}_k\mathbf{Q}_k\mathbf{H}_k^\dagger
			|}+\sum_{j=1}^{k-1}w_j\log\big|
		\mathbf{I}+\mathbf{G}
		\sum_{i=j+1}^{K}\mathbf{Q}_i\mathbf{G}^\dagger
		\big|
		-\lambda{\rm tr}(\mathbf{Q}_k),\\
		% k-cvx
		&f_k^{cvx}(\mathbf{Q}_k,\mathbf{Q}_{\bar{k}})=
		-w_k\log{|
			\mathbf{I}+({
				\mathbf{I}+\mathbf{G}
				\sum_{j=k+1}^{K}\mathbf{Q}_j\mathbf{G}^\dagger
			})^{-1}\mathbf{G}\mathbf{Q}_k\mathbf{G}^\dagger
			|}
		+\sum_{j=1}^{k-1}w_j\log\frac{\left|
			\mathbf{I}+\mathbf{H}_j
			\sum_{i=j}^{K}\mathbf{Q}_i\mathbf{H}_j^\dagger
			\right|}{\left|
			\mathbf{I}+\mathbf{H}_j
			\sum_{i=j+1}^{K}\mathbf{Q}_i\mathbf{H}_j^\dagger
			\right|}\notag \\ 
		&\quad\quad\quad\quad\quad-\sum_{j=1}^{k-1}w_j\log\big|
		\mathbf{I}+\mathbf{G}
		\sum_{i=j}^{K}\mathbf{Q}_i\mathbf{G}^\dagger
		\big|
		%	-\lambda\sum_{j=1}^{k-1}{\rm tr}(\mathbf{Q}_j)-\lambda P
		+\sum_{j=k+1}^{K}w_jR_j-\lambda{\rm tr}(\sum_{j=1}^{K}\mathbf{Q}_j-\mathbf{Q}_k- P).
		\end{align} \label{kreceiverccvcvx1} 
	\end{subequations}
	%	\begin{align}  
	%	\mathbf{Q}^{(i)}_k= {\rm arg}\max \limits_{\mathbf{Q}_k}w_k 
	%	\log{|
	%		\mathbf{I}+({
	%			\mathbf{I}+\mathbf{H}_k
	%			\sum_{j=k+1}^{K}\mathbf{Q}_j\mathbf{H}_k^\dagger
	%		})^{-1}\mathbf{H}_k\mathbf{Q}_k\mathbf{H}_k^\dagger
	%		|}+\sum_{j=1}^{k-1}w_j\log\big|
	%	\mathbf{I}+\mathbf{G}
	%	\sum_{i=j+1}^{K}\mathbf{Q}_i\mathbf{G}^\dagger
	%	\big|
	%	-{\rm tr}[(\lambda\mathbf{I}-\mathbf{A}_{k}^{(i)\dag})\mathbf{Q}_k], \label{reformuP3}
	%	\end{align} 
	% Restore the current equation number.
	\setcounter{equation}{\value{MYtempeqncnt}}
	\setcounter{equation}{12}
	% The IEEE uses as a separator
	\hrulefill
	% The spacer can be tweaked to stop underfull vboxes.
	%\vspace*{4pt}
\end{figure*}

To reach the border of the secrecy capacity region, we need
to solve the WSR problem in \eqref{eq_WSR}. Unfortunately, this problem
is non-convex and challenging to solve–even after finding the
optimal encoding order. In this section, we propose a numerical algorithm to solve \eqref{eq_WSR}. For simplicity of derivations and
notation, without loss of generality, we assume the encoding
order is ascending order, that is $\pi_k=k$. The Lagrangian of
the problem \eqref{eq_WSR} is
\begin{align} \label{lagr}
L(\mathbf{Q}_1, \dots, \mathbf{Q}_K, \lambda) =
\sum_{k=1}^K w_k R_{k}-\lambda(\sum_{k=1}^K{\rm tr}(\mathbf{Q}_k)-P),
\end{align}
where $\lambda$ is the Lagrange multiplier related to the total power constraint.

To solve this, similar to \cite{park2015weighted} and \cite{qi2021signaling} we can then apply the block
coordinate descent (BCD) algorithm. The BCD algorithm  periodically 
finds optimal solutions for a single block of variables while maintaining other blocks of variables fixed at each iteration  \cite{razaviyaynunified}.  {An example of BCD  is the block successive maximization method (BSMM) \cite{park2015weighted} which} updates the covariance matrices by successively optimizing a lower bound of local approximation  of
$f(\mathbf{Q}_1, \dots, \mathbf{Q}_K) = L(\mathbf{Q}_1, \dots, \mathbf{Q}_K, \lambda)$. At  iteration $i$ of the algorithm, the variables $\mathbf{Q}^{(i)}_k$, $k = 1, 2,  \dots, K$, are updated by solving the following problem \cite{razaviyaynunified, park2015weighted}
\begin{align} \label{bsm1}
\mathbf{Q}^{(i)}_k = {\rm arg} \max_{\mathbf{Q}_k \succcurlyeq 0} f(\mathbf{Q}^{(i)}_1, \dots, \mathbf{Q}^{(i)}_{k-1},	\mathbf{Q}_k, \mathbf{Q}^{(i-1)}_{k+1}, \dots,  \mathbf{Q}^{(i-1)}_K),
\end{align}
The function
$f(\mathbf{Q}_1, \dots, \mathbf{Q}_K)$ can be written into the summation of  convex and concave functions
\begin{align} \label{conv_conca}
f(\mathbf{Q}_1, \dots, \mathbf{Q}_K) =f^{ccv}_{k}(\mathbf{Q}_k, \mathbf{Q}_{\bar{k}}) + f^{cvx}_{k}(\mathbf{Q}_k, \mathbf{Q}_{\bar{k}}),
\end{align}
\noindent {in which $\mathbf{Q}_{\bar{k}} = (\mathbf{Q}_1,  \dots, \mathbf{Q}_{k-1}, \mathbf{Q}_{k+1},  \dots, \mathbf{Q}_{K})$ represents all covariance matrices excluding $\mathbf{Q}_k$,}  $f^{ccv}_{k}(\mathbf{Q}_k, \mathbf{Q}_{\bar{k}})$ is the concave function of $\mathbf{Q}_k$ by fixing $\mathbf{Q}_{\bar{k}}$, and $f^{cvx}_{k}(\mathbf{Q}_k, \mathbf{Q}_{\bar{k}})$ denotes the convex function of $\mathbf{Q}_k$, 
The details of the functions are omitted {for space limitations.} We write the functions in \eqref{kreceiverccvcvx1} shown at the bottom of the page. 
	%Interested readers can refer to \cite{park2015weighted, qi2021signaling} for two-user examples.

For the $i$th iteration, the convex function for $f^{cvx}_{k}(\mathbf{Q}_k, \mathbf{Q}^{(i)}_{\bar{k}})$ can be lower-bounded by its gradient \cite{razaviyaynunified}
\begin{align} \label{Talyor}
\setcounter{equation}{16}
f^{cvx}_{k}(\mathbf{Q}_k, \mathbf{Q}^{(i)}_{\bar{k}}) \geq& 	f^{cvx}_{k}(\mathbf{Q}^{(i)}_k, \mathbf{Q}^{(i)}_{\bar{k}}) + {\rm {tr}}[\mathbf{A}_{k}^{(i)\dag}	(\mathbf{Q}_k - \mathbf{Q}^{(i)}_k)]
\end{align}
in which $\mathbf{Q}^{(i)}_{\bar{k}} =( \mathbf{Q}^{(i)}_1, \dots, \mathbf{Q}^{(i)}_{k-1},	\mathbf{Q}^{(i-1)}_{k+1}, \dots, \mathbf{Q}^{(i-1)}_K)$, the first $k-1$ covariance matrices has been optimized in the $i$ iteration, while the $k+1$ to $K$ covariance matrices are from the previous $i-1$ iteration waiting to be optimized. The power price matrix is a partial derivative with respect to $\mathbf{Q}_k$
\begin{align}
\mathbf{A}_{k}^{(i)} =& \triangledown_{\mathbf{Q}_k} f^{cvx}_{k}(\mathbf{Q}_k, \mathbf{Q}_{\bar{k}})|_{\mathbf{Q}^{(i)}_{k}, \mathbf{Q}^{(i)}_{\bar{k}}}.   \label{eq:A1}
\end{align}
A lower bound of the function $f(\mathbf{Q}_1, \dots, \mathbf{Q}_K)$  at $i$th iteration by substituting the right terms in \eqref{Talyor} into \eqref{conv_conca} can be re-written as
\begin{align} \label{reformuP21}
f(\mathbf{Q}^{(i)}_1,\dots, \mathbf{Q}^{(i)}_{k-1},	\mathbf{Q}_k,  \mathbf{Q}^{(i-1)}_{k+1}, \dots, \mathbf{Q}^{(i-1)}_K)  &\geq	 \notag \\
f^{ccv}_{k}(\mathbf{Q}_k, \mathbf{Q}^{(i)}_{\bar{k}})+
f^{cvx}_{k}(\mathbf{Q}^{(i)}_k, \mathbf{Q}^{(i)}_{\bar{k}})
+ {\rm {tr}}[\mathbf{A}_{k}^{(i)\dag}&	(\mathbf{Q}_k - \mathbf{Q}^{(i)}_k)].
\end{align} 
Then, we optimize the right-hand side of the inequality in \eqref{reformuP21} by omitting the constant terms and obtaining the general iteration formula. The general iteration formula is a convex problem and any convex tool can be applied. It is worth mentioning that the solution of maximizing the right-hand side of \eqref{reformuP21} is unilaterally optimal following Nash equilibrium \cite[Proposition~2]{scutari2013decomposition}.  
The details of the algorithm and its complexity analysis are omitted, interested readers can refer to \cite{park2015weighted, qi2021signaling}, for examples.

%\The general iterative solution is derived as \eqref{reformuP3}. \eqref{reformuP3} is a convex problem and any convex tool can be applied. The closed form solution for \eqref{reformuP3} can be achieved by rotaion modeling \cite{vaezi2017journal, zhang2020rotation}.

\section{Simulations}\label{sec:V}
In this section, we present some simulation results through two examples. 

\textbf{Example 1:} We consider the two-user case ($K=2$) and set the channels and power the same as \cite{park2016secrecy}, which are 
\begin{align}
\mathbf{H}_1& =\left[
\begin{matrix} 
1 &  -0.5 \\
0.5 &  2
\end{matrix}\right], \;
\mathbf{H}_2  =  \left[
\begin{matrix} 
-0.3 &  1 \\
2.0 &  -0.4
\end{matrix}\right], \notag \\
\mathbf{G}& =\left[
\begin{matrix}  0.8 & -1.6
\end{matrix}\right],  \; P=1. \notag
\end{align}
We set $w_1 = w_2=0.5$ to verify the subcase of the weighted sum rate, i.e.,  sum-rate.  If user 1 is encoded first and then user 2,  i.e., $\pi_1 = 1$, $\pi_2 = 2$, 
the optimal secrecy rates are $(R_1, R_2) = (0.8334,  0.7643)$  $\rm{nats/sec/Hz}$ and the secrecy sum-rate is $1.5977$.  If user~2 is encoded first and then user~1, i.e., $\pi_1 = 2$, $\pi_2 = 1$, the optimal secrecy rates are $(R_1, R_2) = (0.5324,  1.065)$  which yields the same security sum-rate which is $1.5977$. 
If we set weights  $w_1$ and $w_2$ from 0 to 1 with a step 0.01, the  results of the rate pairs are shown in Fig.~\ref{fig:dpc1}. The blue circle curve denotes the secure broadcasting with ascending coding order, while the yellow circle curve denotes the descending coding order. Different encoding orders result in different rates. The purple and green curves represent the achievable rate region of MIMO BC \cite{weingartens2006capacity}, i.e., no eavesdropper or $\mathbf{G} = \mathbf{0}$. 
%It is shown in \cite{benammar2015secrecy}
%that the outer bound on the secrecy capacity region which, in absence of
%security constraints, reduces to the   MIMO-BC \cite{vishwanath2003duality}. 

The ordering of the users weight  in the WSR
maximization problem determines the optimal encoding order.  An example is shown in Fig.~\ref{fig:order}. If $w_1 > w_2$, the optimal encoding order is to encode 
 message $M_1$ first and message $M_2$ next. If $w_1 \leq w_2$, the optimal encoding order is to encode $M_2$  first and  $M_1$ is encoded next.  If $w_1 = w_2$, the two orders both will give the
 same sum-rate but each order will give a different corner point
 of the capacity region. See Fig.~\ref{fig:dpc1} and Fig.~\ref{fig:order} for illustrative
 representations. 
%For example, the ratio $4:6$ implies user 2 encodes first and then user 1 with $w_1=0.4$ and $w_2=0.6$. 
%Ratio $5:5$ implies a half of the sum-rate where $w_1=w_2=0.5$. 
%Note that the weights only need to satisfy non-negative condition. 
%For simplification, we normalize the sum of all weights to a unit. 

%\section{Simulation}\label{sec:simulation}
\begin{figure}[t]
	\centering
	\begin{minipage}[t]{0.32\textwidth}
		\centering
		\includegraphics[width=6.50cm]{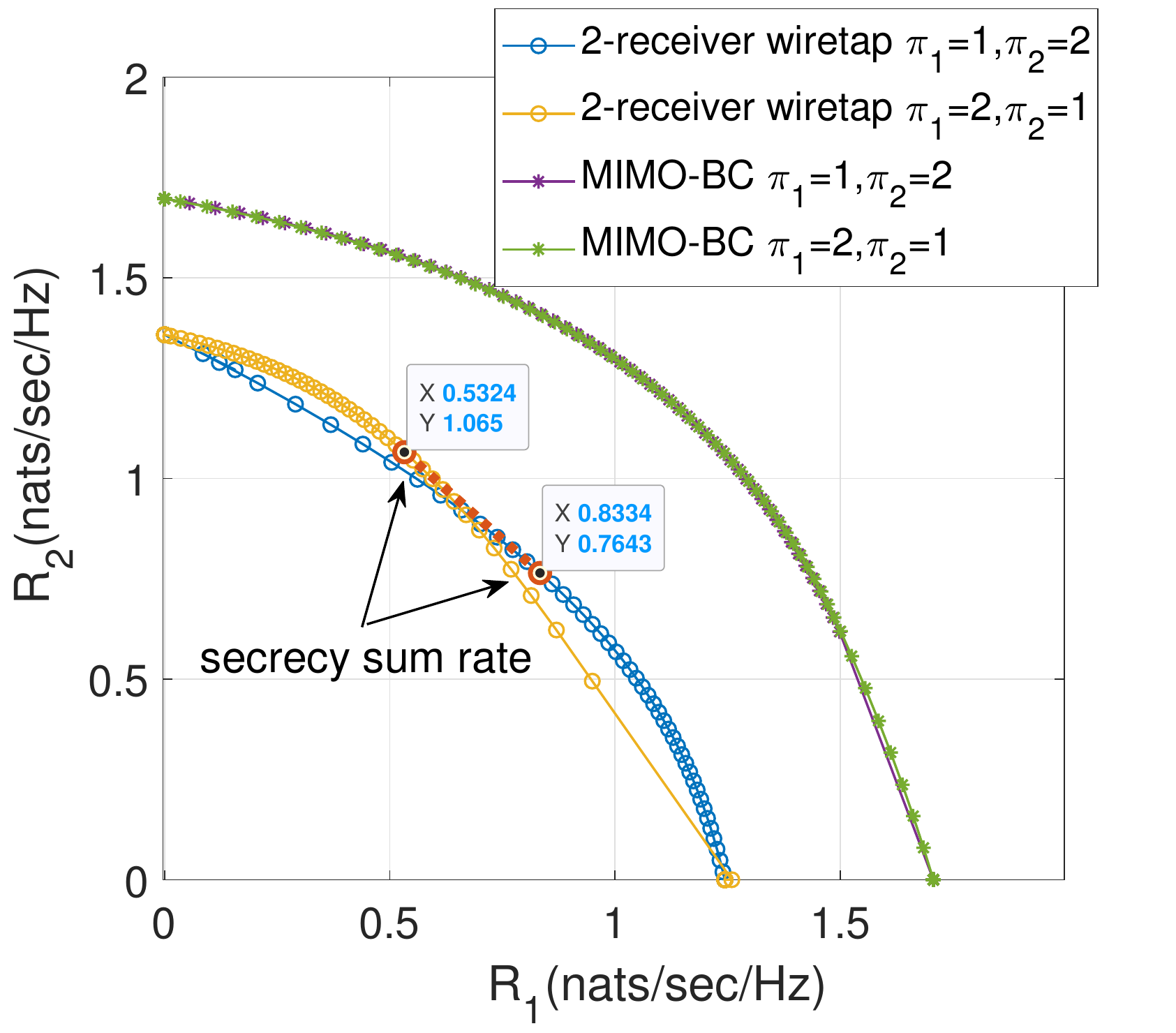}%TworeceiverRisit.eps
		\caption{Secrecy rate region for $K=2$.}
		\label{fig:dpc1}
	\end{minipage} 
	\begin{minipage}[t]{0.32\textwidth}
		\centering
		\includegraphics[width=6.50cm]{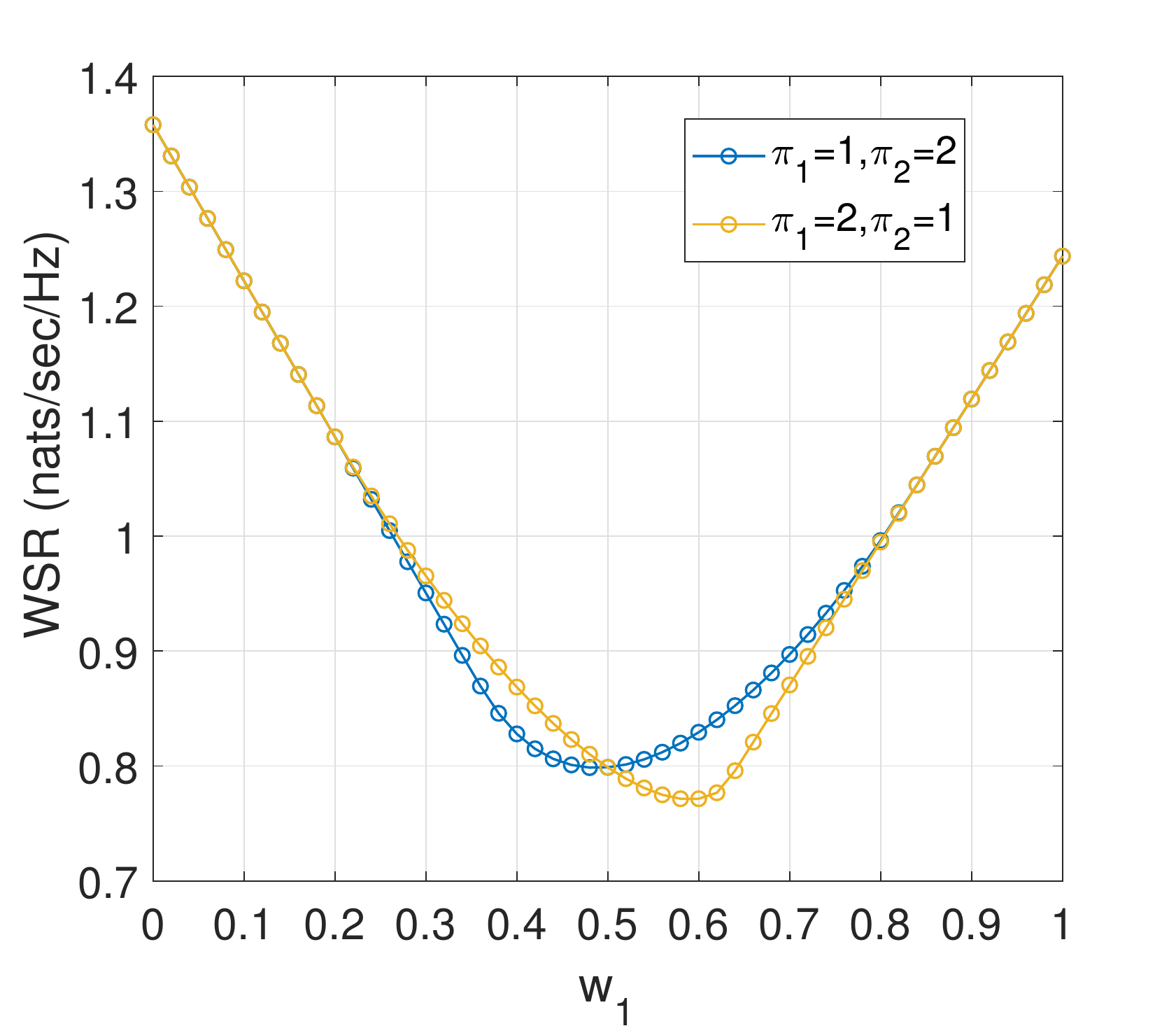}
		\caption{Weights versus WSR for $K=2$. }
		\label{fig:order}
	\end{minipage} 
	\begin{minipage}[t]{0.45\textwidth}
		\centering
		\includegraphics[width=8.0cm]{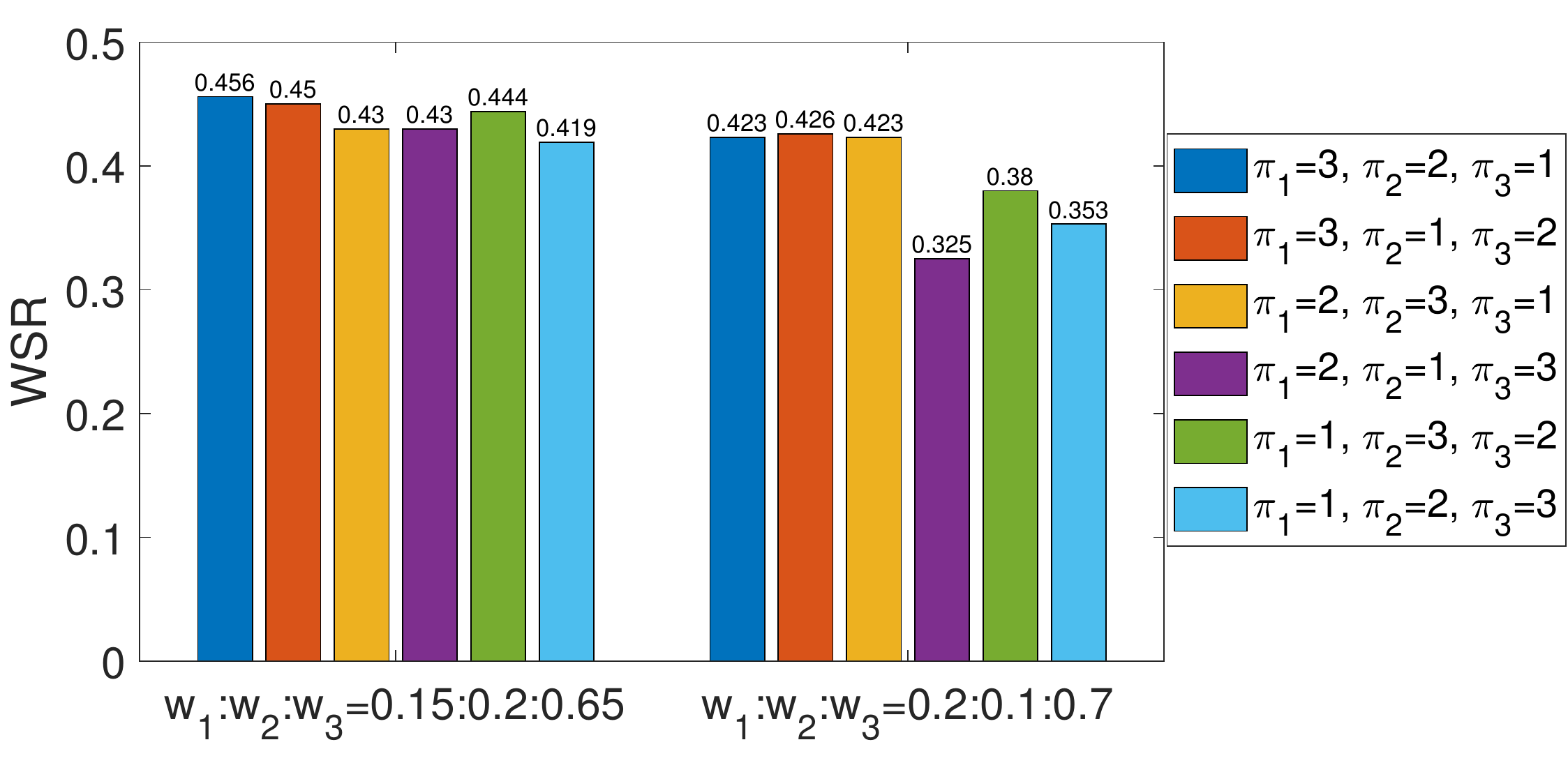}
		\caption{Different weights with all coding order permutations.}
		\label{fig:dpc4}
	\end{minipage} 
\end{figure} 

%\begin{figure}[t]
%	\centering
%	\includegraphics[width=0.32\textwidth]{TwoUserRisit.eps}
%	% where an .eps filename suffix will be assumed under latex,
%	% and a .pdf suffix will be assumed for pdflatex
%	\caption{Secrecy rate region $K=2$.}
%	\label{fig:dpc1}
%\end{figure}
%\begin{figure}[t]
%	\centering
%	\includegraphics[width=0.32\textwidth]{order.eps}
%	\caption{Secrecy rate region for a two-receiver scenario.}
%	\label{fig:order}
%\end{figure}

%\begin{figure}[t]
%	\centering
%	\includegraphics[width=0.45\textwidth]{rateUserKv2.epsnoZF}
%	\caption{Different weights with all coding order permutations.}
%	\label{fig:dpc4}
%\end{figure}

\textbf{Example 2:} We consider $K=3$, in which the channels are 
\begin{align} \label{example2}
\mathbf{H}_1& =\left[
\begin{matrix} 
-0.4332 + 0.7954i &  -0.3152 - 1.8835i \\
-1.0443 + 1.2282i &  -0.2614 + 0.2198i
\end{matrix}\right], \notag \\
\mathbf{H}_2&  =  \left[
\begin{matrix} 
   1.3389 - 0.5995i &  -0.6924 - 0.4542i \\
-1.2542 + 0.1338i &  -2.1644 + 0.6520i
\end{matrix}\right], \notag \\
\mathbf{H}_3&  =  \left[
\begin{matrix} 
 1.0291 - 0.0212i &    -0.3016 - 0.3662i \\
 0.1646 + 0.5179i &    0.3075 + 0.2919i
\end{matrix}\right], \notag \\
\mathbf{G}& = \left[
\begin{matrix} 
  -0.0875 - 0.9443i &   -0.4637 + 0.7799i
\end{matrix}\right]. \notag
\end{align}
\noindent 
Figure~\ref{fig:dpc4} shows two different weights allocations with all six encoding orders.
Each bar reflects the  secure WSR under a specific weighting coefficient. The maximum WSR is achieved with the encoding order $w_{\pi_1} > w_{\pi_2} >   \dots > w_{\pi_K}$. For example, for $(w_1, w_2, w_3)=(0.15, 0.2, 0.65)$ we get $\pi_1 = 3$, $\pi_2 = 2$, $\pi_3 = 1$, whereas for  $(w_1, w_2, w_3)=(0.2, 0.1, 0.7)$ we have $\pi_1 = 3$, $\pi_2 = 1$, $\pi_3 = 2$.  That is, the optimal encoding order is $\pi = [3, 2, 1]$ in the former and
$\pi = [3, 1, 2]$ in the latter. 
%The sum-rate  is $0.7700$ ${\rm nats/sec/Hz}$ for all encoding order when  $(w_1, w_2, w_3)=(1/3, 1/3, 1/3)$.
%In this case, $P=1$, $n_t=2$, $n_1=n_2=2$, and $n_e=1$. Figure~\ref{fig:3users} shows the results of the secrecy achievable rate regions achieved by WSR maximization. 
%If we set any one of the user  idle, the problem reduces to the two-user case. The projection  of the secrecy capacity onto  $(R_1, R_2)$, $(R_2, R_3)$, or $(R_1, R_3)$ plane is the two-user capacity regions.
%\begin{figure}[t]
%	\centering
%	\includegraphics[width=0.38\textwidth]{User3_region.eps}
%	% where an .eps filename suffix will be assumed under latex,
%	% and a .pdf suffix will be assumed for pdflatex
%	\caption{Secrecy rate region for a two-user MIMO-NOMA. }
%	\label{fig:3users}
%\end{figure}

\section{Conclusions}\label{sec:VI}
Optimal encoding order for $K$-receiver  wiretap channel is established in this paper. The proof indicates that the optimal order depends on the
weight ordered. A receiver having a higher weight should be encoded earlier. The proof is characterized by translating
the original WSR maximization problem of the MIMO-BC
wiretap channel into an equivalent problem. This finding reduces the complexity of determining the
capacity region from $K!$ encoding orders to only one-time
encoding. We  have also solved the WSR maximization problem using BSMM. Numerical results verify the optimal encoding order.

\balance
\bibliography{proposal20220117}
\bibliographystyle{ieeetr}

\end{document}